\documentclass[a4paper,14pt]{extarticle}
\usepackage{geometry}
\usepackage[utf8]{inputenc}
\usepackage[utf8]{inputenc}
\usepackage[english]{babel}
\usepackage{amsmath}
\usepackage{amssymb}
\usepackage{MnSymbol}
\usepackage{amsthm}
\usepackage{braket}
\usepackage{amsfonts}
\usepackage{array}
\usepackage{comment}
\usepackage{tikz}
\usetikzlibrary{arrows}
\usetikzlibrary{tikzmark,calc}
\usepackage{pdfpages}
\usepackage{graphicx}
\usepackage{mathrsfs}
\usepackage{calligra}
\usepackage{mathtools}
\usepackage{titlesec}
\usepackage{titlesec}
\titlelabel{\thesection}
\titlelabel{\thesubsection}
\titlelabel{\thesubsubsection}
\titleformat{\section}[block]{\color{black}\large\bfseries\filcenter}{\thesection}{1em}{}[]
\titleformat{\subsection}[block]{\color{black}\large\filcenter}{\thesubsection\quad}{1em}{}[]
\titleformat{\subsubsection}[block]{\color{black}\large\filcenter}{\thesubsubsection\quad}{1em}{}[]
\titleformat{\part}[frame]
  {\bfseries\Huge}
  {\filright\large\enspace{\partname}\enspace}
  {40pt}
  {\Large\filcenter\MakeUppercase}
\newcommand\vertarrowbox[3][6ex]{%
  \begin{array}[t]{@{}c@{}} #2 \\
  \left\uparrow\vcenter{\hrule height #1}\right.\kern-\nulldelimiterspace\\
  \makebox[0pt]{\scriptsize#3}
  \end{array}%
}

\def\mC{{\mathbb C}}
\def\mZ{{\mathbb Z}}

 \newtheorem{theorem}{Theorem}
 
 \newtheorem{definition}{Definition}

\newcommand{\tr}{{\mathtt{tr}}}

\newcommand{\al}{\alpha}

\newcommand{\la}{\lambda}

\newcommand{\ka}{\kappa}

\newcommand{\ga}{\gamma}

\makeatletter
\newcommand*{\rom}[1]{\expandafter\@slowromancap\romannumeral #1@}
\makeatother

\def\tt{\mathtt}
\def\frak{\mathfrak}
\def \f {\mathfrak}

\def \h {\f{h}}
\def \mfg{{\frak g}}

\def \mp {\frak{p}_{2}}

\def \uson {\mathtt{U}(\frak{so}(2n))}
\def \uson1 {\mathtt{U}(\frak{so}(2n+1))}

\def \mgl {{\frak{gl}}}
\def \mfgl11{{\frak{gl}(1|1)}}
\def\sl{{\frak {sl}}}

\def\mfosp{{\frak{osp}}}

\def\mfosp12{{\frak {osp}(1|2)}}

\def\b{{\frak{b}}}
\def \Ml{\tt{M}_{\lambda}}

\def \(,){\left(\cdot,\cdot\right)}

\def\p{\partial}

\def \n {\frak{n}}

\def \V {\mathscr{V}}

\def \M {\mathscr{M}}

\def \mp {\frak{p}}

\begin{document}
\begin{flushright}
    ITEP-TH-8/25
\end{flushright}
\vspace{1cm}
\begin{center}
{\Large{\sf Partial theta-series and branching rules for the  $\sl(3)$ parabolic Verma  modules
}
}\\
\vspace{10mm} {\sf E. Dotsenko$^{\star}$  
 \\ \vspace{7mm}
 \vspace{2mm} $\star$ - {\sf NRC ''Kurchatov Institute'', 123182,  Moscow, Russia}\\

 \vspace{4mm}
 {\footnotesize \sf email: e.ivv.dotsenko@gmail.com\\}
 \vspace{2mm}

\vspace{2mm}}

\end{center}

\begin{abstract}
The monodromy of the $\sl(3)$ Casimir flat connection around root hyperplanes is studied. For the computation of the traces of the root monodromy operators, acting on the parabolic Verma modules, we deduce branching rules w.r.t. the corresponding root $\sl(2)$ subalgebras. We show that the traces of the monodromy operators are partial theta functions of special type. 
\end{abstract}
\vspace{2mm}

\begin{center}
    \textbf{Key wors:} monodromy of the Casimir connection, parabolic Verma modules, branching rules, partial theta functions.
\end{center}
\newpage
\section{Introduction}
In this paper we continue the study of the relation between monodromy operators of the Casimir connection and generalized theta series, initiated in \cite{thetamonodromy}. In the case of $\sl(3)$ Casimir connection we produce partial theta series that correspond to three root monodromy operators. In order to compute these theta series we deduce the branching rules for the parabolic Verma modules. General properties of the branching rules for the parabolic Verma modules were studied in \cite{kobetal}, \cite{kob}. The connection of the present work to the mock modular functions arising in the modern theory of invariants of knots and 3-manifolds \cite{3dmod}, \cite{parkthesis} and \cite{gm} is not developed. Let us recall some facts about Casimir connection. \par 
To a complex semi-simple finite-dimensional Lie algebra $\mfg$ one can associate two flat connections - Knizhnik-Zamolodchikov (KZ) flat connection $\nabla_{\mathtt{KZ}}$ and the Casimir flat connection $\nabla$. The KZ flat connection appears naturally in the 2d WZW model \cite{kz}. The Casimir connection commutes with the KZ connection and shifts the twist parameters of $\nabla_{\mathtt{KZ}}$, see \cite{trigtl}, \cite{fmtv}, \cite{tava}. \par
Historically the theory of the Casimir connection was built by De-Concini, Felder-Markov-Tarasov-Varchenko \cite{fmtv} and independently by Millson and Toledano-Laredo \cite{mtl}, \cite{qgfc}. The generality of \cite{fmtv} allows to consider the Casimir connection associated to Kac-Moody Lie algebra (without Serre relations) $\mfg(\tt{C})$, where $\tt{C}$ is the Cartan matrix of general type. Let us also recall the recent development \cite{t}. \par 
The paper organized as follows: in \textbf{section 2} we recall the $\sl(3)$ Casimir connection and a method of constructing partial theta series, based on branching rules; in \textbf{section 3} we deduce the branching rules for the Borel $\sl(3)$ Verma module, also we compute the corresponding partial theta series; in \textbf{section 4} we deduce the branching rules for the parabolic $\sl(3)$ Verma modules, and the corresponding partial theta series are computed.
\section{Partial theta functions and branching rules}
In this section we recall the definition of $\sl(3)$ Casimir connection, acting in the parabolic Verma module. The connection between branching rules w.r.t. the root $\sl_{\al_{ij}}(2)$ ($\al_{ij}$ - is a positive root of $\sl(3)$) and the spectrum of the monodromy operator $\M_{ij}$ around root hyperplane $\tt{Ker}(\al_{ij})$ is also shown. \par
Let us consider the trivial bundle
\begin{equation}
\begin{gathered}
        \V_{\la}^{\f{p}} = \h^{reg}\times \tt{M}_{\la}^{\f{p}}, \\
\h^{reg} = \h \backslash \bigcup_{1\leq i<j \leq 3} z_{i} = z_{j}, \label{trivialbundle}
\end{gathered}
\end{equation}
where $\tt{M}_{\la}^{\f{p}}$ -  is the parabolic Verma module with the highest weight $\la \in \h^*$ ( we will recall its definition in the sections 3 and 4 of the present paper).
\par On  $\V_{\la}^{\f{p}}$   the Casimir connection is defined as follows 
\begin{equation}
\begin{gathered}
        \nabla = d - \frac{1}{\hbar} \left( \frac{\kappa_{12}}{z_{1}-z_{2}}d\left(z_{1}-z_{2} \right)+ \frac{\kappa_{23}}{z_{2}-z_{3}}d\left(z_{2}-z_{3} \right) + \frac{\kappa_{13}}{z_{1}-z_{3}}d\left(z_{1}-z_{3} \right) \right), \\
        \kappa_{ij} = E_{ij}E_{ji}+ E_{ji}E_{ij},  \\
        d = dz_{1}\p_{z_{1}}+dz_{2}\p_{z_{2}}+ dz_{3}\p_{z_{3}}.
\end{gathered}
\end{equation}

The key property of $\nabla$ is that it is flat
\begin{equation}
    \nabla^2 = 0.
\end{equation}

Section $\psi$ on the bundle $\V_{\la}^{\f{p}}$ is called flat if
\begin{equation}
    \nabla \psi = 0.
\end{equation}
The space of flat sections is denoted as $\Gamma\left(\V_{\la}^{\f{p}}\right)$. There is an action of the fundamental group on the space of the flat sections 
\begin{equation}
    \pi_{1}(\h^{reg}): \Gamma\left(\V_{\la}^{\f{p}}\right)\to \Gamma\left(\V_{\la}^{\f{p}}\right).
\end{equation}
It is expected that the trace of the monodromy operator $\tr_{\Gamma(\V_{\la}^{\f{p}})}\M_{\ga}$ along loop $\ga\in \pi_{1}(\h^{reg})$ generalizes partial Appell-Lerch sums (definition of such function was given, for example, in \cite{thetamonodromy}). 
Computation of the trace of the monodromy operator along general loop $\ga$ is a difficult problem. In this paper we consider the case when the contour $\ga$ is a loop around the hyperplane $z_{i} = z_{j}$. Since the $\nabla$ is regular, then the root monodromy operator is conjugated to
\begin{equation}
    \M_{ij}  \simeq \exp \left( \frac{2\pi \sqrt{-1} \kappa_{ij}}{\hbar} \right).
\end{equation}
Thus the problem of the computation of the spectrum $\M_{ij}$ is reduced to the problem of the computation of the spectrum $\kappa_{ij},$ acting on the Verma module $\tt{M}_{\la}^{\f{p}}$. In order to compute the spectrum of $\kappa_{ij}$ we will deduce the bunching rules w.r.t. the corresponding $\sl_{\al_{ij}}(2)\subset \sl(3)$. Let us recall the definition of branching rules \cite{kobetal}, \cite{kob}  
\begin{definition}
Branching rule of the irreducible $\mfg$ representation $\tt{V}$ w.r.t. Lie subalgebra $\mfg'\subset \mfg$ is the decomposition of $\tt{V}$ as $\mfg'$ representation into the direct sum of irreducible $\mfg'$ representations.
\end{definition}

\par
Let us consider the parabolic Verma module $\tt{M}_{\la}^{\f{p}}$ as a representation of $\sl_{\al_{ij}}(2)$ subalgebra, such representation we denote as $\tt{M}_{\la}^{\f{p}}|_{\sl_{\al_{ij}}(2)}$, then
\begin{equation}
    \tt{M}^{\f{p}}_{\la}|_{\sl_{\al_{ij}}(2)} = \bigoplus_{k\in I^{\la}_{ij}}   \tt{R}_{k}^{ \oplus [ \tt{M}^{\f{p}}_{\la}: \tt{R}_{k}]}, \label{rootdecomposition}
\end{equation}
where $\tt{R}_{k}$\footnote{ In all examples considered in the present paper $\tt{R}_{k}$ is either finite dimensional representation or Verma module. } - is a module over  $\sl_{\al_{ij}}(2)$,  $[\tt{M}^{\f{p}}_{\la}: \tt{R}_{k}]$ - is the multiplicity of $\tt{R}_k$ in $\tt{M}^{\f{p}}_{\la}$, $I^{\la}_{ij}$ is a subset of highest weights of the root $\sl(2)$, it will be specified explicitly in the Theorems 1 and 3.
\par Having the decomposition $\eqref{rootdecomposition}$ it is not very difficult to compute the spectrum of $\kappa_{ij}$, which is determined by the branching rule $\eqref{rootdecomposition}$

Thus the trace of the root monodromy operator has the following shape 
\begin{equation}
    \tt{tr}_{\tt{M}^{\f{p}}_{\la}}\M_{ij} = \sum_{k\in I^{\la}_{ij}} [ \tt{M}^{\f{p}}_{\la}: \tt{R}_{k}] \tt{tr}_{\tt{R}_{k}} \exp{2\pi i \kappa_{ij}}  \label{thetaseriesfrombranchingrule}.
\end{equation}

In the sections 3 and 4 it will be shown that such a trace is a partial theta series of special kind.
\par
\section{Branching rules for the Borel Verma modules}
In this section we derive the branching rules of the Borel Verma module w.r.t. root $\sl_{\al_{ij}}(2)$ subalgebras, and present the partial theta series corresponding to these decompositions. \par
Let us start our consideration with the definition of the positive roots $\tt{R}_+$
\begin{equation}
    \tt{R}_+ = \{\al_{12} = \epsilon_{1}- \epsilon_{2}, \al_{23} = \epsilon_{2}- \epsilon_{3}, \al_{13} = \epsilon_{1}- \epsilon_{3} \}.
\end{equation}
Root subalgebras has the following form
\begin{equation}
    \begin{gathered}
        \sl_{\al_{12}}(2) = \mC\cdot E_{12} \oplus \mC \cdot h_{12} \oplus \mC \cdot E_{21}, \\
        \sl_{\al_{23}}(2) = \mC\cdot E_{23} \oplus \mC \cdot h_{23} \oplus \mC \cdot E_{32}, \\
        \sl_{\al_{13}}(2) = \mC\cdot E_{13} \oplus \mC \cdot h_{13} \oplus \mC \cdot E_{31}, \\
        h_{ij} = E_{ii} - E_{jj}.
    \end{gathered}
\end{equation}
Let us choose the Borel subalgebra in the following way
\begin{equation}
    \b = \begin{pmatrix}
        * & * & * \\
        0 & * & * \\
        0 & 0 & *
    \end{pmatrix}.
\end{equation}
Let us recall that the algebra $\sl(3)$ has the following triangular decomposition
\begin{equation}
\begin{gathered}
        \sl(3) = \n_{-}\oplus \b,\,\,\text{where} \\
    \n_{-} = [\b,\b]^{t}.
\end{gathered}
\end{equation}

\begin{definition}
The Verma module  $\Ml^{\b}$ is induced from the one-dimensional representation of $\b$ -  $\mC_{\la}$, $\la\in \h^{*}$ 
\begin{equation}
    \Ml^{\b} = \tt{U}(\sl(3))\otimes_{\tt{U}(\b)} \mC_{\la} = \tt{U}(\n_{-})\cdot \mC_{\la}.
\end{equation}
\end{definition}
Now let us move on to the derivation of branching rules for root subalgebras.

\begin{theorem}
    The branching rules for the root $\sl_{\al_{ij}}(2)$ subalgebras have the form
    \begin{subequations}
    \begin{equation}
        \Ml^{\f{b}}|_{\sl_{\al_{12}}(2)} = \bigoplus_{n\leq m} \tt{M}_{\langle \al_{12}, \la - n\al_{12}-m\al_{23} \rangle }\label{borelal1branchingrule},
    \end{equation}
          \begin{equation}
               \Ml^{\f{b}}|_{\sl_{\al_{23}}(2)} = \bigoplus_{m \leq n} \tt{M}_{\langle \al_{23}, \la - n\al_{12}-m\al_{23} \rangle } , \label{borelal2branchingrule}
          \end{equation}
         \begin{equation}
         \Ml^{\f{b}}|_{\sl_{\al_{13}}} = \bigoplus_{n,m = 0}^{\infty} \tt{M}_{\langle \al_{13}, \la - n\al_{12}-m\al_{23} \rangle }. \label{borelal1al2branchingrule}
         \end{equation}
    \end{subequations}
\end{theorem}
Let us prove the equality $\eqref{borelal1branchingrule}$ (the rest are proved similarly).
\begin{proof}
Verma module decomposes into the direct sum of its weight subspaces
\begin{equation}
\begin{gathered}
        \Ml^{\f{b}} = \bigoplus_{n,m = 0}^{\infty} \Ml^{\f{b}}\left(\la - n\al_{12} - m\al_{23} \right).
\end{gathered}
\end{equation}
Root generator $E_{12}$ acts as follows 
\begin{equation}
    E_{12}: \Ml^{\f{b}}\left(\la - n\al_{12} - m\al_{23} \right) \to \Ml^{\f{b}}\left(\la - \left(n-1 \right)\al_{12} - m\al_{23} \right).
\end{equation}
Let us describe this action for all $m$ and $n$. Such a description will give us the brunching rules w.r.t. the $\sl_{\al_{12}}(2)$ subalgebra. Let $n\leq m$. Let us choose a basis in $\Ml^{\f{b}}\left(\la - n\al_{12} - m\al_{23} \right)$ in the following way
\begin{equation}
\begin{gathered}
       \tt{x}_{1} =  E_{21}^{n}E_{32}^{m}v_{\la}, \,\,\tt{x}_{2} =  E_{21}^{n-1}E_{32}^{m-1}E_{31}v_{\la},\,\, \ldots, \,\, \tt{x}_{k+1} =  E_{21}^{n-k}E_{32}^{m-k}E_{31}^{k}v_{\la}, \,\, \ldots,\,\, \\ \tt{x}_{n} =  E_{21}E_{32}^{m-n+1}E_{31}^{n-1}v_{\la},\,\,\tt{x}_{n+1} =   E_{32}^{m-n}E_{31}^{n}v_{\la}.
\end{gathered}
\end{equation}
Basis in  $\Ml^{\f{b}}\left(\la - \left(n-1 \right)\al_{12} - m\al_{23} \right)$ we will choose similarly as follows
\begin{equation}
    \begin{gathered}
        \tt{y}_{1} =  E_{21}^{n-1}E_{32}^{m}v_{\la}, \,\, \tt{y}_{2} = E_{21}^{n-2}E_{32}^{m-1}E_{31}v_{\la},\,\, \ldots, \,\, \tt{y}_{k+1} =  E_{21}^{n-k-1}E_{32}^{m-k}E_{31}^{k}v_{\la}, \,\, \ldots,\,\, \\ \tt{y}_{n-1} =  E_{21}E_{32}^{m-n+2}E_{31}^{n-2}v_{\la},\,\, \tt{y}_{n} =  E_{32}^{m-n+1}E_{31}^{n-1}v_{\la}.
    \end{gathered}
\end{equation}
The root generator acts as follows
\begin{equation}
    \begin{gathered}
        E_{12}\tt{x}_{1} = n\left(\la_{1}+m+1-n\right)\tt{y}_{1}, \\
        E_{12}\tt{x}_{2} = \left(n-1\right)\left(\la_{1}+m-n\right)\tt{y}_{2} - \tt{y}_{1}, \\
        \vdots \\
        E_{12}\tt{x}_{k+1} = \left(n-k\right)\left(\la_{1}+m+1-k-n\right)\tt{y}_{k+1} - k\tt{y}_{k}, \\
        \vdots \\
        E_{12}\tt{x}_{n} = \left(\la_{1}+m+2-2n\right)\tt{y}_{n} - \left(n-1\right)\tt{y}_{n-1},\\ 
         E_{12}\tt{x}_{n+1} = -n \tt{y}_{n}, \,\,\, \text{where} \\
         \la_{1} = \langle \al_{12},\la \rangle. \label{borelvermaal1branching}
    \end{gathered}
\end{equation}
From the formulas  $\eqref{borelvermaal1branching}$ one sees that for the general  $\la_{1}$ vectors $E_{12}\tt{x}_{1}, \ldots, E_{12}\tt{x}_{n}$ generate $$\Ml^{\f{b}}\left(\la - \left(n-1 \right)\al_{12} - m\al_{23} \right),$$
from which it follows that $\tt{ker}\left(E_{12} \right)$  one dimensional. 
\par Let us consider the case when $m< n$. Let us prove that the map 
\begin{equation}
    E_{12}: \Ml^{\f{b}}\left(\la - n\al_{12} - m\al_{23} \right) \to \Ml^{\f{b}}\left(\la - \left(n-1 \right)\al_{12} - m\al_{23} \right)
\end{equation} is a bijection. Let us consider the following basis in $\Ml^{\f{b}}\left(\la - n\al_{1} - m\al_{2} \right)$ 
\begin{equation}
    \begin{gathered}
       \tt{x}_{1} =  E_{21}^{n}E_{32}^{m}v_{\la},\,\, \tt{x}_{2} = E_{21}^{n-1}E_{32}^{m-1}E_{31}v_{\la},\,\, \ldots, \tt{x}_{k+1} =  E_{21}^{n-k}E_{32}^{m-k}E_{31}^{k}v_{\la}, \,\,\ldots , \\
       \tt{x}_{m} = E_{21}^{n-m+1}E_{32}E_{31}^{m-1}v_{\la}, \,\, \tt{x}_{m+1} = E_{21}^{n-m}E_{31}^{m}v_{\la}.
    \end{gathered}
\end{equation}
Let us also choose a basis in $\Ml^{\f{b}}\left(\la - \left(n-1\right)\al_{12} - m\al_{23} \right)$ 
\begin{equation}
    \begin{gathered}
        \tt{y}_{1} =  E_{21}^{n-1}E_{32}^{m}v_{\la},\,\, \tt{y}_{2} = E_{21}^{n-2}E_{32}^{m-1}E_{31}v_{\la},\,\, \ldots, \tt{y}_{k+1} =  E_{21}^{n-k-1}E_{32}^{m-k}E_{31}^{k}v_{\la}, \,\,\ldots , \\
       \tt{y}_{m} = E_{21}^{n-m}E_{32}E_{31}^{m-1}v_{\la}, \,\, \tt{y}_{m+1} = E_{21}^{n-m-1}E_{31}^{m}v_{\la}.
    \end{gathered}
\end{equation}
The action of  $E_{12}$ in this pair of bases has the following form
\begin{equation}
    \begin{gathered}
        E_{12}\tt{x}_{1} =n\left(\la_{1}+m+1-n\right)\tt{y}_{1}, \\
        E_{12}\tt{x}_{2} =\left(n-1\right)\left(\la_{1}+m-n\right) \tt{y}_{2}+ \tt{y}_{1}, \\
        \vdots \\
        E_{12}\tt{x}_{k+1} =  \left(n-k\right)\left(\la_{1}+m+1-k-n\right)\tt{y}_{k+1}- k \tt{y}_{k}, \\
        \vdots \\
        E_{12}\tt{x}_{m+1} =\left(n-m\right)\left(\la_{1}+1-n\right)\tt{y}_{m+1} - m \tt{y}_{m}
    \end{gathered}
\end{equation}
From the equations above one sees that the vectors $E_{12}\tt{x}_{1},\, \ldots , E_{12}\tt{x}_{m+1} $ for a general $\la_{1}$ generate $$\Ml^{\f{b}}\left(\la - \left(n-1\right)\al_{12} - m\al_{23} \right),$$ thus $E_{12}$ is a bijection. This finishes the proof.
\end{proof}\par 
\begin{theorem}
One has the following
    \begin{equation}
    \begin{gathered}
               \tr_{\Ml^{\f{b}}} \M_{13} = \sum_{k\geq 0} \frac{q^{- 2k^2 + \left(2k+1\right)\left(\la_{1}+\la_{2}\right)}}{\left(1-q^{-\left(2k+1 \right)}\right)^2}. \label{boreltracemal1al2}
    \end{gathered}
    \end{equation}
\end{theorem}
In turn the traces $ \tr_{\Ml^{\f{b}}} \M_{12},\,\, \tr_{\Ml^{\f{b}}} \M_{23}$ are divergent.
\begin{proof}

Indeed, substituting $\eqref{borelal1al2branchingrule}$ in $\eqref{thetaseriesfrombranchingrule}$ and computing the spectrum of $\ka_{13}$ on each $\sl_{\al_{13}}(2)$ Verma module we will gain the following partial theta series
\begin{equation}
\begin{gathered}
     \tr_{\Ml^{\f{b}}} \M_{13}= \sum_{n, m \geq 0}\tr_{\tt{M}_{\langle \al_{13}, \la - n\al_{12}-m\al_{23}\rangle}} \exp\left(\frac{2\pi \sqrt{-1}  \kappa_{13}}{\hbar}\right) =  \sum_{n, m\geq 0}\sum_{k \geq 0}q^{\left(2k+1\right)\left(\la_{1} + \la_{2} -n - m \right) -2k^{2}} = \\ = 
     \sum_{k = 0}^{\infty} \frac{q^{-2k^2 + \left( 2k+1 \right)\left(\la_{1} + \la_{2}\right)}}{\left(1 -q^{-2k-1} \right)^2}.
\end{gathered}
\end{equation}
\end{proof}
\subsection{Regularization of the divergent series}
The traces $\tr_{\Ml^{\f{b}}} \M_{12},\,\, \tr_{\Ml^{\f{b}}} \M_{23}$ are divergent. One can regularize them by inserting an additional operator under the trace 
\begin{equation}
    \begin{gathered}
        \tr^{reg}_{\Ml^{\f{b}}} \M_{12} = \tr_{\Ml^{\f{b}}} \M_{12} t_{1}^{h_{1}}t_{2}^{h_{2}}, \\
        \tr^{reg}_{\Ml^{\f{b}}} \M_{23} = \tr_{\Ml^{\f{b}}} \M_{23} t_{1}^{h_{1}}t_{2}^{h_{2}}
    \end{gathered}
\end{equation}
\begin{theorem}
    The following equalities hold
    \begin{subequations}
        \begin{equation}
        \begin{gathered}
                        \tt{tr}^{reg}_{\Ml^{\b}}\M_{12} =  t_{1}^{\la_{1}}t_{2}^{\la_{2}} \sum_{k = 0}^{\infty}\frac{\left(t_{1}^{-2}t_{2} \right)^kq^{\la_{1}\left(2k+1 \right)-2k^2}}{\left(1- t_{1}^{-2}t_{2}q^{-2\left(2k+1 \right)}\right)\left(1-t_{1}t_{2}^{-2}q^{2k+1}\right)} - \\ - \frac{\left(t_{1}^{-2}t_{2} \right)^{k+1}q^{\left(\la_{1}-2 \right)\left(2k+1 \right)-2k^2}}{\left( 1-t_{1}^{-2}t_{2}q^{-2\left(2k+1\right)}\right)\left(1-t_{1}^{-1}t_{2}^{-1}q^{-2k-1}\right)},
        \end{gathered}
        \end{equation}
        \begin{equation}
        \begin{gathered}
            \tt{tr}^{reg}_{\Ml^{\b}}\M_{23} = t_{1}^{\la_{1}}t_{2}^{\la_{2}}\sum_{k = 0}^{\infty}\frac{\left(t_{1}t_{2}^{-2} \right)^{k}q^{\la_{2}\left(2k+1 \right)-2k^2}}{\left(1-t_{1}t_{2}^{-2}q^{-2(2k+1)} \right)\left(1-q^{2k+1}t_{1}^{-2}t_{2} \right)} - \\ 
            - \frac{\left(t_{1}t_{2}^{-2}\right)^{k+1}q^{\left(\la_{2}-2\right) \left(2k+1 \right)-2k^2}}{\left(1-t_{1}t_{2}^{-2}q^{-2(2k+1)} \right)\left(1-t_{1}^{-1}t_{2}^{-1}q^{-(2k+1)} \right)}.
        \end{gathered}
        \end{equation}
    \end{subequations}
\end{theorem}
\begin{proof}
    Direct computation.
\end{proof}
\section{Branching rules for the parabolic Verma modules}
In this section we recall the definition of the parabolic Verma module, also we deduce the branching rules w.r.t. all root $\sl_{\al_{ij}}(2)$ and compute the partial theta series that correspond to these branching rules.  \par
Let us consider the following parabolic subalgebra $\mp\subset \sl(3)$
\begin{equation}
\mp =  \begin{pmatrix}
									* & * & *  \\
                                        0 & * & * \\
									0 & * & *  
\end{pmatrix}.
\end{equation} 
This subalgebra has the decomposition
\begin{equation}
\begin{gathered}
        \mp = \f{l}\oplus \f{n}_{+},  \\
\f{l} = \mgl(1) \oplus \sl(2),\\
\f{n}_{+} = \mC E_{12}\oplus \mC E_{13}.
\end{gathered}
\end{equation} 
Let us also recall the decomposition of the $\sl(3)$ algebra 

\begin{equation}
\begin{gathered}
       \sl(3) =  \f{n}_{-}\oplus \mp, \\
 \f{n}_{-} = \mC E_{21}\oplus \mC E_{31} = \f{n}_{+}^{\tt{t}}
\end{gathered}
\end{equation}
\begin{definition}
Parabolic Verma module is an induced module
\begin{equation}
    \tt{M}(\tt{V}) = \tt{U}(\sl(3))\otimes_{\tt{U}(\frak{p})}\tt{V},
\end{equation}
where $\tt{V}$  is the finite-dimensional representation of $\f{p}$, which is induced from the finite dimensional representation of Levi subalgebra $\f{l}\subset \f{p}$.
\end{definition}
Let us consider the examle of finite-dimensional module $\tt{V}_{\la}$ over $\f{l}$. The Cartan subalgebra $\h\subset \f{l}$ acts on the highest weight vector $v_{\la}\in \tt{V}_{\la}$ as follows
\begin{equation}
    \begin{gathered}
       h_{1} v_{\la} = \la_{1}v_{\la}, \\
       h_{2} v_{\la} = \la_{2}v_{\la},
    \end{gathered}
\end{equation}
where $\la_{1}\in \mC \backslash \mZ$, $\la_{2}\in \mZ_{\geq 0}$. $\tt{V}_{\la}$ has the form
\begin{equation}
 \tt{V}_{\la}  = \bigoplus_{i = 0}^{\la_{2}} \mC E_{32}^{i} v_{\la}.
\end{equation}
Let us describe how the corresponding parabolic Verma module is structured. Since $\f{n}_{-}$ is commutative, the Verma module has the form
\begin{equation}
    \tt{M}(\tt{V}_{\la}) = \mC[E_{21}, E_{31}] \tt{V}_{\la}. 
\end{equation}
It is not difficult to see that the character of the parabolic Verma module is 
\begin{equation}
    \chi\left(\tt{M}\left(\tt{V}_{\la}\right)\right) = \tt{tr}_{\tt{M}\left(\tt{V}_{\la}\right)} e^{\phi_{1}h_{1} + \phi_{2}h_{2}} = t_{1}^{\la_{1}}t_{2}^{\la_{2}}\frac{\sum_{ i = 0}^{\la_{2}}t_{1}^{i}t_{2}^{-2i}}{(1-t_{1}^{-2}t_{2})(1-t_{1}^{-1}t_{2}^{-1})}, \label{paraboliccharacter} 
\end{equation}
where $t_{1} = e^{\phi_{1}},\, t_{2} = e^{\phi_{2}}$.

\begin{theorem}
Branching rules for the parabolic Verma modules w.r.t. the root subalgebras  $\sl_{\al_{ij}}(2)$ have the following form
    \begin{subequations}
    \begin{gather}
                \tt{M}\left(\tt{V}_{\la} \right)|_{\sl_{\al_{12}}(2)} = \bigoplus_{s = 0}^{\infty} \bigoplus_{r = 0}^{\la_{2}}\tt{M}_{\la_{1} +r -s}, \label{al1parabolicbranchingrule} \\
                \tt{M}\left(\tt{V}_{\la} \right)|_{\sl_{\al_{23}}(2)} = \bigoplus_{i = 0}^{\la_{2}}\tt{L}_{i}^{\oplus\left(i+1\right)}\oplus\bigoplus_{i = \la_{2}+1}^{\infty}\tt{L}_{i}^{\oplus \left(\la_{2}+1\right)},  \label{al2parabolicbranchingrule} \\
                 \tt{M}\left(\tt{V}_{\la} \right)|_{\sl_{\al_{13}}(2)} = \bigoplus_{s = 0}^{\infty}\bigoplus_{r = 0}^{\la_{2}}\tt{M}_{\la_{1}+\la_{2} - r-s}. \label{al1al2parabolicbranchingrule}
    \end{gather}
    \end{subequations}
\end{theorem}
The proofs of the branching rules in the Theorem above are similar to each other. Let us prove the equality $\eqref{al1parabolicbranchingrule}$ as an example.

\begin{proof}
Let us start the proof of the branching rule
$\eqref{al1parabolicbranchingrule}$ with the consideration of the following weight subspace  $\tt{M}\left(\tt{V}_{\la} \right)\left(\la + m\al_{12} + k\al_{13}\right)$, where $m\geq 1$ and $k\geq 0$. By the character formula $\eqref{paraboliccharacter}$ it has the dimension 
 \begin{equation}
   \dim  \tt{M}\left(\tt{V}_{\la} \right)\left(\la + m\al_{12} + k\al_{13}\right) = \begin{cases}
       k+1,\,\, \text{if}\,\, k\leq \la_{2}\\
       \la_{2} + 1,\,\, \text{if}\,\, k > \la_{2}
   \end{cases}
 \end{equation}
  Let us choose the following basis in this weight space
    \begin{equation}
    \begin{gathered}
               \tt{x}_{1}= E_{21}^{m}E_{31}^{k}v_{\la},\, \tt{x}_{2} = E_{21}^{m+1}E_{31}^{k-1}E_{32}v_{\la}, \ldots , \\ \tt{x}_{j+1} = E_{21}^{m+j}E_{31}^{k-j}E_{32}^{j}v_{\la},\ldots, \tt{x}_{k+1} = E_{21}^{m+k}E_{32}^{k}v_{\la}.
    \end{gathered}
    \end{equation}
    In the weight subspace  $\tt{M}\left(\tt{V}_{\la} \right)\left(\la + \left(m-1\right)\al_{12} + k\al_{13}\right)$ choose a basis
    \begin{equation}
       \begin{gathered}
               \tt{y}_{1}= E_{21}^{m-1}E_{31}^{k}v_{\la},\, \tt{y}_{2} = E_{21}^{m}E_{31}^{k-1}E_{32}v_{\la}, \ldots , \\ \tt{y}_{j+1} = E_{21}^{m+k}E_{31}^{i-1-k}E_{32}^{k}v_{\la},\ldots, \tt{y}_{k+1} = E_{21}^{m+k-1}E_{32}^{k}v_{\la}.
    \end{gathered}
    \end{equation}
    The map between them $$E_{12}:\tt{M}\left(\tt{V}_{\la} \right)\left(\la + m\al_{12} + k\al_{13}\right) \to \tt{M}\left(\tt{V}_{\la} \right)\left(\la + \left(m-1\right)\al_{12} + k\al_{13}\right)$$ in the selected pair of bases has the form
    \begin{equation}
    \begin{gathered}
                E_{12}\tt{x}_{1} = m\left(\la_{1} + 1 - m- k \right) \tt{y}_{1} - k\tt{y}_{2}, \\
        E_{12}\tt{x}_{2} = \left(m+1\right)\left(\la_{1} + 2 - m- k \right)\tt{y}_{2} - \left(k-1\right) \tt{y}_{3}, \\
        \vdots  \\
       E_{12} \tt{x}_{j+1} =  \left(m+j\right)\left(\la_{1} + j+1 - m- k \right)\tt{y}_{j+1} - \left( k-j \right) \tt{y}_{j+2}, \\
       \vdots \\
        E_{12}\tt{x}_{k+1} = \left(m+k\right)\left(\la_{1} +1 - m \right)\tt{y}_{k+1}.
    \end{gathered}
    \end{equation}
    From the formulas above one sees that the linear span of 
    $E_{12}\tt{x}_{1},\ldots,E_{12}\tt{x}_{k+1}$ coincides with $$\tt{M}\left(\tt{V}_{\la} \right)\left(\la + \left(m-1\right)\al_{12} + k\left(\al_{12} + \al_{23}\right)\right),$$ thus $E_{12}$ is a bijection and thus in the weight subspaces under consideration there are no singular vectors.  \par
    Now let us consider the weight subspaces $\tt{M}\left(\tt{V}_{\la} \right)\left(\la -l\al_{23} - k \al_{13}\right)$, $0 \leq l\leq \la_{2}$ and $k\geq 0$. From the character formula $\eqref{paraboliccharacter}$ follows that 
    \begin{equation}
        \dim \tt{M}\left(\tt{V}_{\la} \right)\left(\la -l\al_{23} - k\al_{13}\right) = \begin{cases}
            k+1,\,\,\text{when}\,\, k\leq \la_{2} - l\\
            \la_{2}-l+1,\,\,\text{when}\,\, k>  \la_{2} - l
        \end{cases}
    \end{equation}
To derive the branching rules, let us consider the map 
\begin{equation}
    E_{12}:\tt{M}\left(\tt{V}_{\la} \right)\left(\la -l\al_{23} - k \al_{13} \right)\to \tt{M}\left(\tt{V}_{\la} \right)\left(\la -\left(l+1\right)\al_{23} - \left(k-1\right)\al_{13}\right) \label{E12parabolicmappingsingularvectors}.
\end{equation}
In the space $\tt{M}\left(\tt{V}_{\la} \right)\left(\la -l\al_{23} -  k\al_{13} \right)$ we will choose the basis
\begin{equation}
    \begin{gathered}
        \tt{x}_{1} = E_{32}^{l}E_{31}^{k}v_{\la}, \tt{x}_{2} =E_{21} E_{32}^{l+1}E_{31}^{k-1}v_{\la}, \ldots, \tt{x}_{j+1} = E_{21}^{j}E_{32}^{l+j}E_{31}^{k-j}v_{\la}, \\
        \ldots, \tt{x}_{k} = E_{21}^{k-1}E_{32}^{l+k-1}E_{31}v_{\la}, \tt{x}_{k+1} = E_{21}^{k}E_{32}^{l+k}v_{\la}.
    \end{gathered}
\end{equation}
In the space $\tt{M}\left(\tt{V}_{\la} \right)\left(\la -\left(l+1\right)\al_{23} - \left(k-1 \right) \al_{13} \right)$ we will choose the basis
\begin{equation}
    \begin{gathered}
        \tt{y}_{1} = E_{32}^{l+1}E_{31}^{k-1}v_{\la}, \tt{y}_{2} =E_{21} E_{32}^{l+2}E_{31}^{k-2}v_{\la}, \ldots, \tt{y}_{j+1} = E_{21}^{j}E_{32}^{l+j+1}E_{31}^{k-j-1}v_{\la}, \\
        \ldots, \tt{y}_{k-1} = E_{21}^{k-2}E_{32}^{l+k-1}E_{31}v_{\la}, \tt{y}_{k} = E_{21}^{k-1}E_{32}^{l+k}v_{\la}.
    \end{gathered}
\end{equation}
In our selected pair of basis the map $\eqref{E12parabolicmappingsingularvectors}$ has the following form 
\begin{equation}
    \begin{gathered}
        E_{12}\tt{x}_{1} = -k\tt{y}_{1}, \\
        E_{12}\tt{x}_{2} = \left(\la_{1} + l-k+2 \right)\tt{y}_{1}-\left(k-1\right)\tt{y}_{2}, \\
        \vdots \\
        E_{12}\tt{x}_{j+1} = j\left(\la_{1} + l+j-k+1 \right)\tt{y}_{j}-\left(k-j\right)\tt{y}_{j+1}, \\
         \vdots \\
          E_{12}\tt{x}_{k} = \left(k-1 \right)\left(\la_{1} + l \right)\tt{y}_{k-1}-\tt{y}_{k}, \\
           E_{12}\tt{x}_{k+1} = k\left(\la_{1} + l+1 \right)\tt{y}_{k}.
    \end{gathered}
\end{equation}
From the formulas above follows that the linear span of the vectors 
$E_{12}\tt{x}_{1},\ldots, E_{12}\tt{x}_{k}$ coincides with$$\tt{M}\left(\tt{V}_{\la} \right)\left(\la -\left(l-1\right)\al_{23} - \left(k-1\right) \al_{13} \right).$$ Thus in each  $\tt{M}\left(\tt{V}_{\la} \right)\left(\la -l\al_{23} - k \al_{13}\right)$ there is only one $\sl_{\al_{12}}(2)$ singular vector.   
\end{proof}
\begin{theorem}
The following equalities hold
    \begin{subequations}
    \begin{gather}
            \tt{tr}_{\tt{M}\left(\tt{V}_{\la} \right)}\M_{12} = \sum_{k\geq 0}q^{-2k^2 - \left(2k+1 \right)\la_{1}}\frac{1-q^{-\left(2k+1\right)\left(\la_{2}+1 \right)}}{\left( 1-q^{-\left(2k+1 \right)} \right)\left(1-q^{\left(2k+1 \right)} \right)}, \\
            \tt{tr}_{\tt{M}\left(\tt{V}_{\la} \right)}\M_{23} = \sum_{i = 0}^{\la_{2}}\sum_{k = 0}^{i+1}\left(i+1\right)q^{-2k^2 + \left(2k+1 \right)i} + \sum_{i = \la_{2}+1}^{\infty}\sum_{k = 0}^{i+1}\left(\la_{2}+1\right)q^{-2k^2 + \left(2k+1 \right)i}, \\
             \tt{tr}_{\tt{M}\left(\tt{V}_{\la} \right)}\M_{13} = \sum_{ k = 0}q^{-2k^2 + \left(\la_{1} + \la_{2} \right) (2k+1)}\frac{1-q^{-\left(\la_{2}+1 \right)\left(2k+1 \right)}}{\left(1-q^{-2k-1} \right)^{2}}.
    \end{gather}
\end{subequations}
\end{theorem}
The proofs of the equalities above are of the same type. As an example, we will give the proof of the first of them.
\begin{proof}
    According to $\eqref{thetaseriesfrombranchingrule}$ the deduced branching rule allows to compute the trace of the monodromy operator in the following way
    \begin{equation}
        \begin{gathered}
            \tt{tr}_{\tt{M}\left(\tt{V}_{\la} \right)}\M_{12} = \sum_{r = 0}^{\la_{2}}\sum_{s = 0}^{\infty}\tr_{\tt{M}_{\la_{1} +r -s}}\exp\left(2\pi i \kappa_{12} \right) =\\ = \sum_{k\geq 0}q^{-2k^2 - \left(2k+1 \right)\la_{1}}\frac{1-q^{-\left(2k+1\right)\left(\la_{2}+1 \right)}}{\left( 1-q^{-\left(2k+1 \right)} \right)\left(1-q^{\left(2k+1 \right)} \right)}
        \end{gathered}
    \end{equation}
\end{proof}

 \section{Conclusion}
 In the paper the branching rules for the parabolic Verma modules were deduced and the corresponding partial theta series were presented. The study of modular properties of these theta functions will be given in the subsequent papers. \par
\vspace{10mm}
\textbf{Acknowledgments.} 
The author thanks M. Olshanetsky for his interest in the work and support.  \par

\end{document}